\documentclass{amsproc}

\usepackage{amsfonts}
\usepackage{amssymb}
\usepackage[dvips]{graphics}
\usepackage{wrapfig}
\usepackage{subcaption} 
\usepackage{epsfig}
\usepackage{euscript}
\usepackage{color}
\usepackage{fancyvrb}
\usepackage{fullpage}
\usepackage{preamble_AMS}

\newlength{\imagewidth}
\newtheorem{theorem}{Theorem}[section]
\newtheorem{lemma}[theorem]{Lemma}
\newtheorem{corollary}[theorem]{Corollary}

\theoremstyle{definition}

\theoremstyle{remark}

\theoremstyle{assumption}
\newtheorem{assumption}[theorem]{Assumption}

\numberwithin{equation}{section}

\newcommand{\abs}[1]{\lvert#1\rvert}

\begin{document}

\title{Note on linear response for interacting Hall insulators}

\author{Sven Bachmann}
\address{Department of Mathematics \\ University of British Columbia \\ Vancouver, BC V6T 1Z2 \\ Canada}
\email{sbach@math.ubc.ca}

\author{Alex Bols}
\address{ Instituut Theoretische Fysica, KU Leuven  \\
3001 Leuven  \\ Belgium }
\email{alexander.bols@kuleuven.be}

\author{Wojciech De Roeck}
\address{ Instituut Theoretische Fysica, KU Leuven  \\
3001 Leuven  \\ Belgium }
\email{wojciech.deroeck@kuleuven.be}

\author{Martin Fraas}
\address{ Instituut Theoretische Fysica, KU Leuven  \\
3001 Leuven  \\ Belgium }
\curraddr{Department of Mathematics \\ Virginia Tech \\ Blacksburg, VA 24061-0123 \\ USA}
\email{fraas@vt.edu}

\date{\today}

\dedicatory{}


\begin{abstract}  
We relate explicitly the adiabatic curvature-in flux space- of an interacting Hall insulator with nondegenerate ground state to various linear response coefficients, in particular the Kubo response and the adiabatic response.  The flexibility of the setup, allowing for various driving terms and currents, reflects the topological nature of the adiabatic curvature. We also outline an abstract  connection between Kubo response and adiabatic response, corresponding to the fact that electric fields can be generated both by electrostatic potentials and time-dependent magnetic fields.
Our treatment fits in the framework of rigorous many-body theory, thanks to the gap assumption. 
\end{abstract}

\maketitle

\section{Introduction}

The Hall conductance is given by an \emph{adiabatic curvature},  related to the threading of two Aharonov-Bohm fluxes.  This insight originated with Niu, Thouless and Wu \cite{Thouless85}, see also the work of Avron and Seiler in \cite{AvronSeiler85}. Over the past years, it inspired a mathematically rigorous proof  \cite{HastingsMichalakis} by Hastings and Michalakis of quantization in the integer quantum Hall effect in the many-body context. 

 The goal of this note is not to sketch these developments but rather to review why the adiabatic curvature is indeed a  Hall response coefficient. This is  hence not a new insight, but we found it quite useful to phrase it in the language of modern many-body theory, using tools like quasi-adiabatic evolution and the like. 
 
A related question that one might want to see clarified is the rigorous justification of linear response \emph{per se}. While in general this remains an important problem of mathematical physics, it is under control in the case of Hall responses (exactly because these are non-dissipative responses), see  \cite{bru2016microscopic,Giuliani:2016gn,bachmann2018adiabatic}. This issue will however not be discussed here. 
%
%
%

\section{Setup}\label{sec:Setup}
\subsection{Spaces and operators} \label{sec: spaces}
We use very heavily the setup and notation from a recent paper of ours, namely \cite{bachmann2018quantization}.  We consider a two dimensional discrete torus $\Gamma=\Gamma_L=\bbZ^2_L$ with $\bbZ_L=\bbZ/(L\bbZ)$.  We take $L$ large and even and we often identify $\Gamma$  with the square $ \{ (x_1,x_2) \in \bbZ^2 \; : \;  -L/2 \leq x_{1,2}\leq L/2\}$,  with the appropriate identification of boundary points. 

A finite-dimensional Hilbert space $\bbC^{n}$ is associated to each site $x\in \Gamma$ and there is a preferred basis in $\bbC^n$ labelled by $\sigma$ (as an example, one can think of the $z$-spin number). 
We consider the fermionic Fock space $\caH=\caH_{\Gamma}$ built on the one-particle space $l^2(\Gamma,\bbC^n)$.
The algebra of operators $\caB(\caH)$ is generated by the creation/annihilation operators $c_{x,\sigma}/c^*_{x,\sigma}$: 
$$
\{c_{x,\sigma},c^*_{x',\sigma'} \}=\delta_{x,x'} \delta_{\sigma,\sigma'}, \qquad   \{c^{\sharp}_{x,\sigma},c^{\sharp}_{x',\sigma'} \}= 0
$$
where $\{A,B\}=AB+BA$ and $c^\sharp$ can be either $c$ or $c^*$. 
Any operator $O$ can be written in a unique way as a sum of normal-ordered monomials in $c^{\sharp}_{x,\sigma}$ which are at most of first degree in each $c^{\sharp}_{x,\sigma}$.  Referring to this unique representation, we write  $O_S$ for the `restriction to $S$', namely the sum of monomials in $O$ containing only $c^{\sharp}_{x,\sigma}$ with $x\in S$.  

For obvious reasons, we call $S$ the `spatial support' of $O_S$.   Also, we will consider only Hamiltonians and observables that are in the \emph{even} subalgebra,  i.e.\ they contain only monomials of even degree.   A direct consequence of this is that, for even $O,O'$ we have $[O_S,O'_{S'}]=0$ whenever $S\cap S'=\emptyset$. 
An oft-used operator is the  particle number at $x$, given by $n_x=\sum_{\sigma}   c^*_{x,\sigma}c_{x,\sigma}$ and the particle number in $X$, given by $n_X=\sum_{x\in X} n_x$. 

We will in general write $S^r$ for the neighborhood 
\begin{equation} \label{def: neighborhood}
S^r= \{ x|\, \mathrm{\dist}(x,S) \leq r\}
\end{equation}
Here the distance $\mathrm{\dist}(\cdot,\cdot)$ refers to the Euclidian distance on the underlying continuous torus $[-L/2,L/2]^2$ with opposite edges identified. 

Consider an observable $O_L$ on $\Gamma_L$ whose support $S$ fits inside a smaller square, say $|x_{1,2}| \leq L/4$ for all $x\in S$, then we can define a corresponding $O_{L'}$ on $\Gamma_{L'}$ for $L'>L$ by the identification of $\Gamma_L$ with a square. 
This realizes a natural embedding of $\caB(\caH_{\Gamma_L})$ into $\caB(\caH_{\Gamma_{L'}})$.  We  will use this to fix an observable $O$ and consider it implicitly for all (sufficiently large) $L$. For example Assumption \ref{tl limit} relies on this construction.

We also need another class of operators, representing Hamiltonians, currents, etc.  They are of the type $G=\sum_{X \in \Gamma} G_X$, with
\begin{enumerate} 
\item  $G_X=0$ unless $\mathrm{diam}(X) \leq R$ for some fixed range $R<\infty$.
\item   $\norm{G_X} \leq m$ for some fixed $m$. 
\end{enumerate}
For lack of a better name,  we call (the $L$-sequence of) $G$ a `local Hamiltonian' whenever the above conditions are satisfied for all $L$ with $m,R$ independent of $L$. 
Of course, one can devise a framework\footnote{The literature on mathematical statistical physics uses the framework of `interaction potentials', see e.g.\ \cite{simon2014statistical}} to consider `the same' $G$ for different $L$, but we will not need this explicitly. 
\subsection{The Hamiltonian}\label{sec: hamiltonians}
Our framework allows to consider rather arbitrary local Hamiltonians, but for the sake of simplicity, we restrict to a class with nearest neighbour hopping: 
$$
H = \sum_{\sigma,\sigma'}\sum_{x\sim x'} \alpha(x,\sigma,x',\sigma') c^*_{x,\sigma} c^{}_{x',\sigma'}  +   \sum_{X \subset \Gamma} B_X
$$
where $x\sim x'$ indicates that $x,x'$ are adjacent, and
\begin{enumerate}
\item  $\alpha(x,\sigma,x',\sigma')= \overline{\alpha(x',\sigma',x,\sigma)}$ to ensure Hermiticity. 
\item  $\sum_X B_X$ is a `local Hamiltonian' as defined above in Section \ref{sec: spaces}.
\item  All $B_X$ are Hermitian  and $[B_X, n_x]=0$ for any $x, X$.
\end{enumerate}
The conserved charge is 
$N=\sum_{x} n_{x}$, 
 i.e.\ for simplicity we assume unit charge per fermion. By $\mathrm{iii})$, we see that the $B_X$ don't contribute to charge transport.
The  natural choice for these $B_X$ is
$$B_{\{x\}} =  \mu n_x  + U n_x^2,\qquad  B_X=0 \text{ if $|X|>1$} $$ 
i.e.\  the Hubbard model with on-site interaction $U$ and  chemical potential $\mu$.
The main assumptions on the Hamiltonian are
\begin{assumption}\label{tilde gap}
$H$ has a non-degenerate ground state $\Psi$ separated from the rest of the spectrum by a distance ${g} >0$, uniformly in the size $L$.   
\end{assumption}
Let us write
$\omega(\cdot)=\langle\Psi,\cdot \Psi \rangle$ for the ground state expectation. Sometimes, as in the upcoming assumption, we need to recall that everything depends on $L$, so we may write $\omega(\cdot)=\omega_L(\cdot)$. 
\begin{assumption}\label{tl limit}
The ground state has a thermodynamic limit in a weak sense: for any observable $O$ with finite support, the limit $\lim_{L\to\infty}\omega_L(O)$ exists. (We used the identification in Section \ref{sec: spaces} of observables for different $L$ to give meaning to $\omega_L(O)=\omega_L(O_L) $)
\end{assumption}
These assumptions are assumed to hold throughout our text and we do not repeat them. That being said, Assumption \ref{tl limit} is only necessary for Lemma \ref{lem: existence of lr} and Theorem \ref{thm: quantization}.
In all what follows, we always mean that error terms, constants $C$, etc can be taken bounded independently of $L$. \\

\subsubsection{Example: interacting Harper model}
We take $n=1$, i.e.\ spinless fermions, so we omit the label $\sigma$. The hopping amplitudes $\alpha$ are specified as
\begin{equation}\label{def: harper}
\alpha(x,x') = \begin{cases} t \ep{\pm\iu  \Phi_L x_1} & \qquad  x_1=x'_1  \quad \text{and}\quad   (x'_2-x_2) \mathrm{mod} L=\pm 1   \\ 
t & \qquad    x_2=x'_2  \quad \text{and}\quad  (x'_1-x_1) \mathrm{mod} L=\pm 1     \end{cases} 
\end{equation}
where $\Phi_L \in 2 \pi \bbZ / L$ is the magnetic flux per unit cell and $t \in \bbR$ is the hopping strength.  Note that $\Phi_L \in 2 \pi \bbZ / L$ ensures that the hopping amplitudes are well defined on the $L \times L$ torus. The infinite volume Harper model\cite{Hofstadter76} is well-defined for all values of the flux $\Phi$ and Lesbegue a.e.\ $\Phi$ satisfy the following property: there is an open set $U\ni\Phi$ and a chemical potential $\mu$ such that, for every $\Phi'\in U$, $\mu$ lies outside of the spectrum of the Harper Hamiltonian.  Let $\Phi$ satisfy this property, then we can find a sequence of fluxes $\Phi_L \rightarrow \Phi$ such that the corresponding sequence of finite-volume Harper models statisfies assumptions \ref{tilde gap} and \ref{tl limit}. So far the non-interacting model. Persistence of gaps for weak interactions  was proven in \cite{hastings2017stability,DeRoeck2018} and also implicitly in \cite{Giuliani:2016gn}, and existence of the thermodynamic limit is standard in this context.


\subsection{Fluxes}\label{subsec:ham with flux}

\subsubsection{One-forms on $\Gamma$}
We want to `thread magnetic fluxes' through the loops of the torus $\Gamma$. These fluxes will be modelled using vector potentials, which we describe as discrete one-forms, \ie objects that can be integrated along oriented paths. The elements of an oriented path are the oriented edges which it traverses. A one-form is a function $A : \Gamma^{e} \rightarrow \bbR$ on the oriented edges of $\Gamma$ such that $A(e)$ flips sign if the orientation of $e$ is reversed. We write $\norm{A}:= \sup_{e}|A(e)|$.

The integral of $A$ along $\gamma$ is then
$$
\int_{\gamma} A := \sum_{e \in \gamma} A(e).
$$

Any function $\theta : \Gamma \rightarrow \bbR$ defines a one-form $\dd \theta$ by $\dd \theta \big( (x, y) \big) = \theta(y) - \theta(x)$. See the appendix for more details on discrete one-forms.

\subsubsection{Hamiltonian with vector potential}

Vector potentials are one-forms $A$.
A background vector potential $A$ is implemented by modifying the Hamiltonian in the following way:
$$
H \to H_A,\qquad \alpha(x,\sigma;x',\sigma') \to \alpha(x,\sigma;x',\sigma') \ep{\iu A \big( (x, x') \big) }.
$$
In practice, we do not need any additional\footnote{Such a flux might be included in the original Hamiltonian, see e.g.\  the Harper model in \eqref{def: harper}} magnetic fluxes piercing the lattice, so we will mostly restrict to vortex-free $A$ i.e.\ $\oint_\gamma A =0$ across loops $\gamma$ that are contractible to a point. 
 The implementation of a vector potential of the form $\dd \theta$ for some function $\theta: \Gamma \to \bbR$ amouts to a \emph{gauge transformation}
\begin{equation}\label{eq:H under gauge}
U_\theta H_A U^*_\theta = H_{A+\dd \theta}
\end{equation}
where
$$
U_\theta=\ep{\iu \langle \theta ,n \rangle },\qquad \langle \theta ,n \rangle\equiv \sum_x \theta(x) n_x.
$$
Consider now a one-form $A$ that is {exact} in the region $\Sigma\subset\Gamma$. By this we mean that $\oint_\gamma A=0$ for any $\gamma$, {not necessarily contractible}, consisting of oriented edges in $\Sigma^e$ (edges whose both vertices are in $\Sigma$). 
Then there exists a function $\theta$, with support in $\Sigma$, such that  
$$A\big|_{\Sigma^e}=  {\dd} \theta. $$
This in particular implies that
\begin{equation} \label{eq: ha locally gauge}
(H_{A})_{\Sigma}=  ( U_{\theta}H  U^*_{\theta})_{\Sigma}.
\end{equation}
If we identify gauge equivalent vector potentials, then there are only two independent nonzero vortex-free classes. A representant of the first (second) class is given by the vector potential $\xi_1$ ($\xi_2$) which takes the value $1/L$ on edges pointing in the positive $1$-direction ($2$-direction), and vanishes on edges pointing in the $2$-direction ($1$-direction). The point is that \emph{locally} the one-form $\xi_i$ is given by $\dd x_i/L$.

Let $\gamma_1,\gamma_2$ be two loops that wind around the torus across the lines $x_2=0,x_1=0$, respectively. Then
$$
\int_{\gamma_i} \xi_j = \delta_{ij}.
$$
Any vector potential of the form $\phi_1 (\xi_1 + \dd \theta_1)+ \phi_2 (\xi_2 + \dd \theta_2)$ describes magnetic fluxes $(\phi_1,\phi_2)$ threaded through the torus, with no magnetic fields on the torus, see Figure ~\ref{fig:fluxes}

\begin{figure}[h]
\centering
\includegraphics[width=0.3\textwidth]{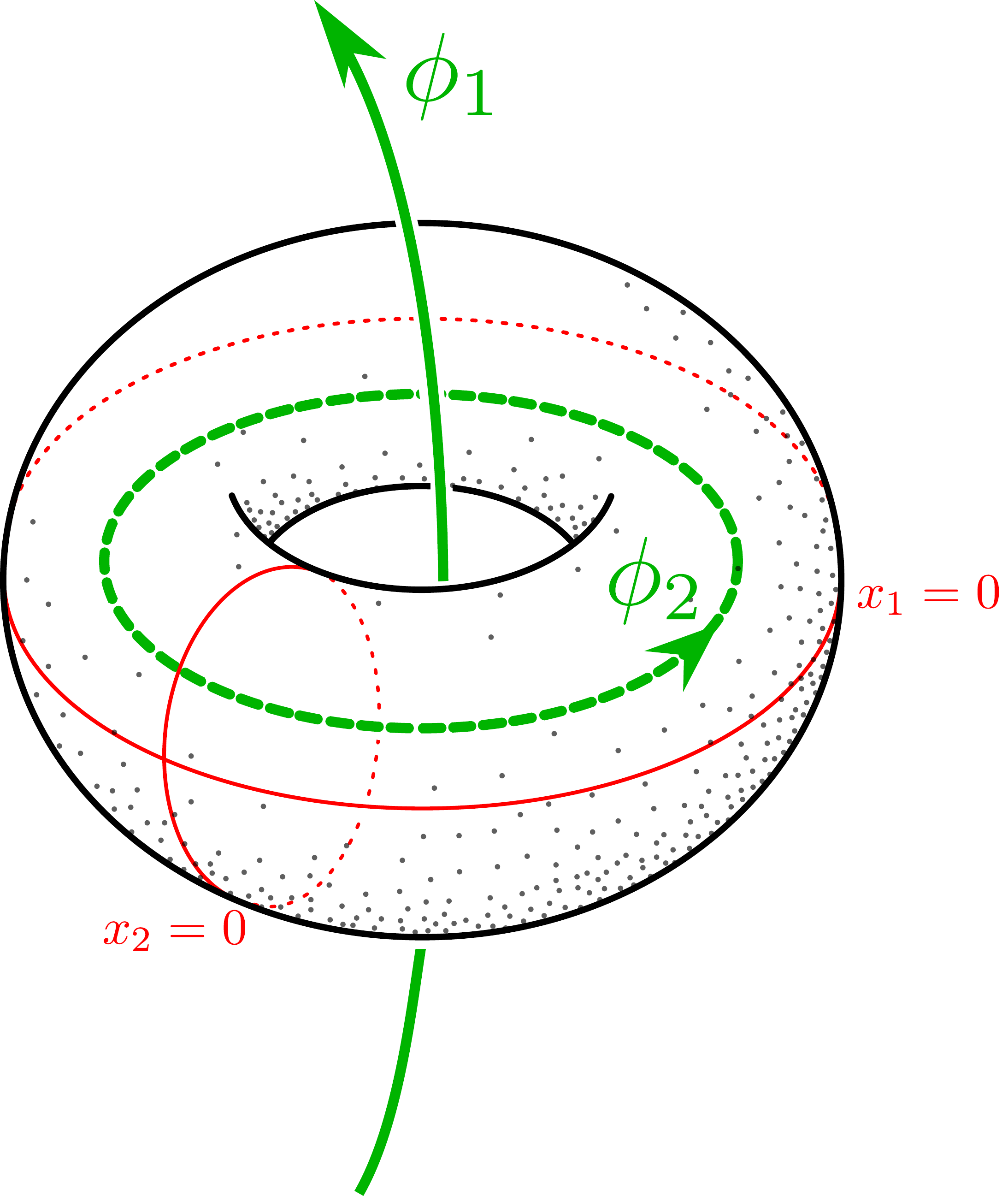}
\caption{The torus $\Gamma$ with threaded fluxes $\phi_1$ and $\phi_2$.}
\label{fig:fluxes}
\end{figure}

\subsection{Current operators}\label{subsec:current operators}

Let us define current, related to the flow of the conserved charge $N$.  
For any connected region $X\subset \Gamma$, the instantaneous change of $n_X$ is given by 
$$
J_{\partial X} = \iu [H,n_X]
$$ 
%

\begin{figure}
\centering
\captionsetup[subfigure]{width=\imagewidth,justification=raggedright}%
  \begin{subfigure}[b]{0.4\textwidth}
      \includegraphics[width=0.9\textwidth]{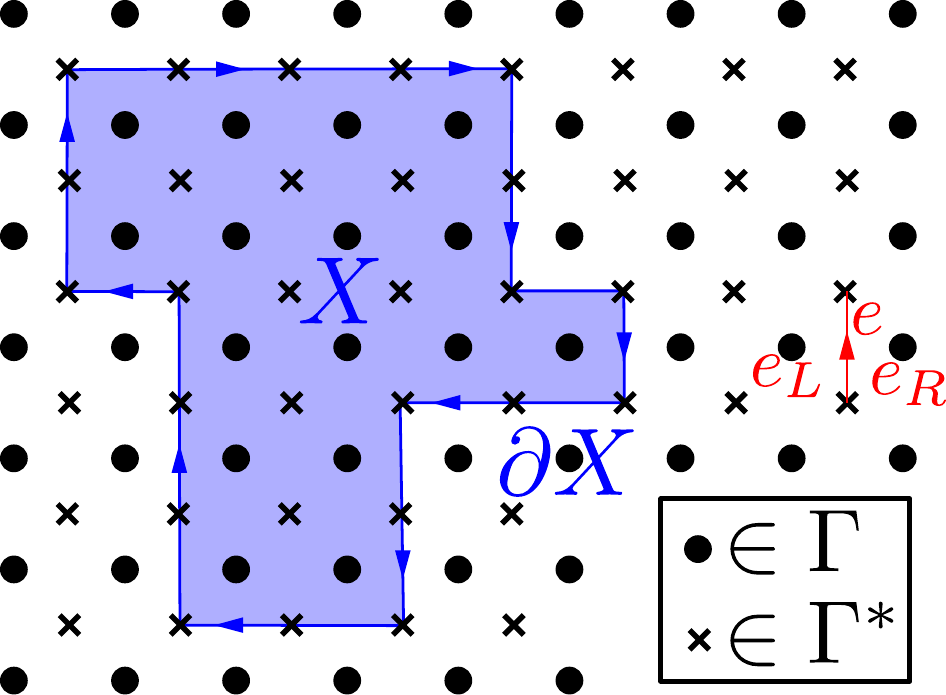}
      \caption{The boundary of the set $X$ as an oriented path, and an edge $e$ together with the vertices $e_L$ and $e_R$ which it passes.}
      \label{fig:set_boundaries_edges}
  \end{subfigure}
  \begin{subfigure}[b]{0.4\textwidth}
      \includegraphics[width=0.8\textwidth]{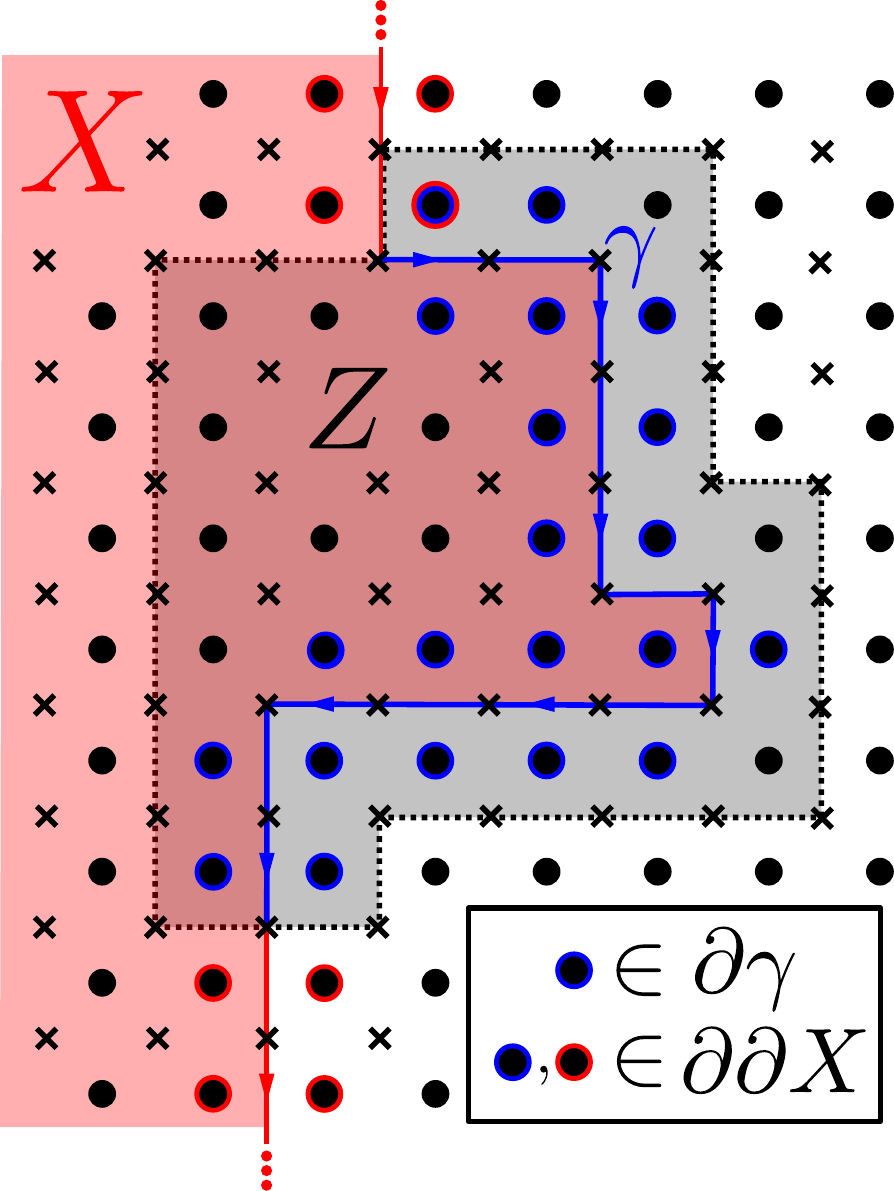}
      \caption{The path $\gamma$ as a restriction of the boundary of $X$.}
      \label{fig:current as region current}
  \end{subfigure}
  \caption{}
\end{figure}

and so it is natural to interpret $J_{\partial X}$ as the current operator through the non-intersecting oriented loop $\partial X$ in the dual lattice $\Gamma^*$. By convention, we orient $\partial X$ in a `counter-clockwise' fashion, \ie when walking along $\partial X$, one sees the set $X$ to the right, see Figure \ref{fig:set_boundaries_edges}.
Moreover, since $J_{\partial X}$ is a sum of local operators situated in a close vicinity of $\partial X$, we can also associate in a natural way a current $J_\gamma$ to every oriented subpath $\gamma$ of $\partial X$.

The ambiguity in doing this amounts to an operator of norm at most $C(R)$ at the ends of $\gamma$, with $R$ the range of the Hamiltonian (actually, only the range of the hopping term would enter here).   Therefore, $J_\gamma$ will be meaningful whenever $|\gamma| \gg C(R)$.  For the sake of explicitness, we give a possible choice. Note first that each oriented edge $e$ of $\Gamma^*$ is uniquely specified by giving the site $e_L$, which lies just to the left of $e$, and the site $e_R$, which lies just to the right of $e$. We set
$$
J_\gamma =   - \iu \sum_{e \in \gamma} \left( \alpha(e_L, e_R) c_{e_R}^* c_{e_L} - \alpha(e_R, e_L) c_{e_L}^* c_{e_R} \right).
$$

The formalism of gauge transformations offers us a handy way to write $J_\gamma$.  Write $\partial\gamma := \cup_{e \in \gamma} \{ e_L, e_R \}$ for the vertices passed by $\gamma$. 
The idea is to find a region $X$ such that $\gamma$ is a subpath of $\partial X$ and to write
$J_\gamma$ as a (spatial restriction of) the current into $X$, i.e.\footnote{this formula might be confusing.  The subscript $Z$ was defined canonically in Section \ref{sec: spaces} as a restriction to a spatial region $Z \subset \Gamma$. In contrast,  $J_{\gamma}$ is simply the current associated to the path $\gamma$ in $\Gamma^*$.}
$$(J_{\partial X})_{Z}= J_{\gamma}$$
for a region $Z$ that selects exactly the right part of $\partial X$. To be precise, $Z$ has to satisfy  $Z \supset \partial \gamma$ and  $(\partial\partial X \setminus \partial \gamma) \cap Z=\emptyset$, see Figure \ref{fig:current as region current}.  

Therefore we have also
\begin{equation} \label{eq: current as region current}
J_\gamma =   \partial_\phi ( \ep{\iu \phi n_X} H \ep{\iu \phi n_X}  )_{Z}  \big|_{\phi=0} = \iu ([  n_X,  H] )_{Z}
\end{equation}
which relates current operators to flux threading.  The last equality is a consequence of the fact that the restriction to a spatial region is a linear map on operators. 

%
%


\section{Response coefficients}

\subsection{Kubo Linear response} \label{subsec:linear response}
The Kubo linear response coefficient $\chi_{J,V}(\nu) $ at frequency $\nu$, describes the response of an observable $J$ to adding a perturbation $\ep{\iu \nu t} V$ to the Hamiltonian starting at $t=0$  \cite{Kubo:1957cl}. We simply start from the well-known expression for the response coefficient:
\begin{equation}\label{eq:finite volume LR}
\chi_{J,V}(\nu) := \iu \lim_{\epsilon \rightarrow 0^+} \int_0^{\infty} \dd t \; \omega \left( [V(-t), J] \right)  \ep{\iu \nu t-\epsilon t},\qquad \text{and $\chi_{J,V} =\chi_{J,V}(0) $}
\end{equation}
where $V(t)= \ep{\iu t H} V \ep{-\iu t H}$.
We should immediately add that it is often crucial to take the thermodynamic limit $L\nearrow \infty$ \emph{before} taking $\epsilon \rightarrow 0^+$. However, for gapped systems (as we are considering) these limits commute:
\begin{lemma}\label{lem: existence of lr}
Let $|\nu| \leq g/2$ (recall that $g$ is the spectral gap). If both $J,V$ are operators with finite support, then
\begin{equation} \label{eq:linear response coefficient}
 \iu \lim_{\epsilon \rightarrow 0^+} \lim_{L\to\infty}  \int_0^{\infty} \dd t \;  \omega \left( [V(-t), J] \right) \ep{\iu \nu t-\epsilon t}
\end{equation}
exists and equals the $L\nearrow \infty$ limit of  \eqref{eq:finite volume LR}. 
\end{lemma}
We will hence consider always \eqref{eq:finite volume LR}  but we stress that the commutativity of limits exhibited in Lemma \ref{lem: existence of lr} actually precludes\footnote{Indeed, let $f_L(t)=\omega_L \left( [V(-t), J] \right)$ and assume that $\lim_{L\to\infty} f_L(t)$ exists and is an integrable function $f(t)$.  Then $\chi(\nu)=2\pi \iu \hat f(\nu)+ \caP \int \dd \nu' \frac{\hat{f}(\nu')}{\nu'-\nu}$,  with $\caP\int \ldots$ denoting the principal part. 
The real and imaginary part are sometimes also called the `dissipative part' and the `reactive part' of the response. 
However, taking the other order of limits, we find that either the imaginary part is zero, or the limit does not exist.} any dissipative effect.

One of the features of the Kubo response that we will rely on, is its \emph{locality}, made explicit in the following lemma. 
\begin{lemma}\label{lem: locality of lr}
Let $|\nu| <g/2$.
Let $J, V$ be local Hamiltonians in the sense of Section \ref{sec: spaces} and assume that the region $Z$ is the intersection of their supports.     Then
\begin{equation}
\chi_{J,V}(\nu)- \chi_{J_{Z^r},V_{Z^r}}(\nu) = \caO(r^{-\infty})
\end{equation}
with $Z^r$ as defined in \eqref{def: neighborhood}.  Actually, $\chi_{J,V}= \iu \: \omega([\caI(J),V])$ (with $\caI$ defined in Section \ref{sec: preliminaries}) which renders this locality explicit. 
\end{lemma}
%

\subsection{Adiabatic response}\label{subsec:adiabatic driving}

Since this setup is less familiar to most readers, we sketch how it is derived from fundamental considerations. 
Consider a family of Hamiltonians
$H_s$ for $s\in [-1,1]$ with uniformly gapped groundsates $\Psi_s$. We require that the map $s\mapsto H_s$ is smooth and that $\partial^n_s H_s=0$ at $s=-1$, for all $n$. 
Now, to put ourselves in the adiabatic regime, the parameter $s$ is varied slowly: the Hamiltonian in physical time $t$ is given by $H^\epsilon(t) := H_{\epsilon t}$. Write $\Psi^\epsilon_{t}$ for the solution to the time-dependent Schr{\"o}dinger equation (TSE)
$$
\iu \: \partial_t \Psi^\epsilon_{t} = H^\epsilon(t) \Psi^\epsilon_{t},   \qquad \text{with initial condition}\quad  \Psi^\epsilon_{-1/\epsilon}=\Psi_{-1}.
$$   
The adiabatic response of some local observable $J$ at parameter $s$ is then defined as the difference between the solution of the TSE and the instantaneous ground state:
\begin{equation}\label{def:adiabatic response}
\chi^{\mathrm{ad}}_{J,H_s} := \lim_{\epsilon \rightarrow 0} \tfrac{1}{\epsilon} \left({\langle \Psi^\epsilon_{s/\epsilon}, J \; \Psi^\epsilon_{s/\epsilon} \rangle - \langle \Psi_{s}, J \; \Psi_{s} \rangle}\right).
\end{equation}
Of course, the same remark about the thermodynamic limit as in Section \ref{subsec:linear response} applies here and we do not comment on that further.  From now on, we will always choose $s=0$ in \eqref{def:adiabatic response}.
Let us give now heuristically evaluate \eqref{def:adiabatic response}. Since the state is close to the instantaneous ground state, let us pretend that they are exactly equal at time $t=0$ and evaluate the difference at $t\gg 1$, but $t$ not growing with $\epsilon$, such that $t/\epsilon$ still morally corresponds to taking $s=0$ in \eqref{def:adiabatic response}. In other words we look at  
$$
\lim_{t\to \infty} \; \lim_{\epsilon \rightarrow 0} \tfrac{1}{\epsilon} \left({\langle \Psi^\epsilon_{t}, J \; \Psi^\epsilon_{t} \rangle - \langle \Psi_{0}, J \; \Psi_{0} \rangle}\right),   \qquad \text{started from}\quad   \Psi^\epsilon_{0} :=\Psi_{0}
$$
 The advantage of doing so is that only the values of $H_s$ near $s=0$ seem to matter.   We expand $H^\epsilon(t)$ around $t=0$, where $H^\epsilon(0)=H_0$, obtaining 
\begin{equation}
H^\epsilon(t)=H_0+ (H^\epsilon(t)-H_0) \approx H_0 + \epsilon t W,\qquad W=\partial_s H_s\big|_{s=0}.
\end{equation}
What we have gained is that the setup now looks very much like the setup of the Kubo response formula: We start at $t=0$ in the ground state and we switch on a time-dependent driving, with the time-dependence being linear.  In this setup one derives the Kubo response formula by making a Dyson expansion of the dynamics, up to first order in $\epsilon$ and taking $t\to \infty$ (after introducing a regularization $\ep{-\delta t}$).  Doing this, we arrive at 
\begin{equation}\label{eq:LR form of adiabatic response}
\chi^{\mathrm{ad}}_{J,H_s} = \iu \lim_{\delta \rightarrow 0^+} \int_0^{\infty} \dd t \; t \omega \left( [W(-t), J] \right) \ep{-\delta t}.
\end{equation}
By standard Fourier techniques and renaming\footnote{We used $\delta$ above to avoid confusion with the unrelated $\epsilon$ in \eqref{def:adiabatic response}.} $\delta\to \epsilon$, this leads to
\begin{equation}\label{eq:adiabatic response is derivative of linear response}
\chi^{\mathrm{ad}}_{J,H_s} =  \lim_{\epsilon \rightarrow 0^+} \frac{\partial}{\partial \nu}  \int_0^{\infty} \dd t   \; \ep{\iu \nu t} \omega \left( [W(-t), J] \right) e^{-\epsilon t}\big|_{\nu=0} = -\iu  \frac{\partial}{\partial \nu} \chi_{J,W}(\nu)\big|_{\nu=0}
\end{equation}
Therefore, this heuristic treatment suggests that the adiabatic response is directly related to the Kubo linear response. 
Indeed, using the adiabatic perturbation theory in \cite{bachmann2017adiabatic}, we have
\begin{lemma}\label{mhw}
Assume that Assumption \ref{tilde gap} holds for all $s \in [-1,1]$ uniformly, that $H_s$ are local Hamiltonians whose parameters $m,R$ can be chosen uniformly in $s$ and that all local terms $H_{s,X}$ are smooth, uniformly in $X$. Then \eqref{eq:LR form of adiabatic response} holds true. 
\end{lemma}

%
%
%

\subsection{Adiabatic curvature}\label{subsec:adiabatic curvature}
We recall the vector potential $A=\phi_1 \xi_1+\phi_2 \xi_2$
introduced in Section \ref{subsec:ham with flux}, corresponding to threaded fluxes $\phi=(\phi_1,\phi_2) $. We now consider the so-called \emph{twist} Hamiltonians
$$
H(\phi)=H(\phi_1,\phi_2)= H_{\phi_1 \xi_1+\phi_2 \xi_2},\qquad \phi \in \bbT^2.
$$
For small $\phi$, the twist Hamiltonian $H(\phi)$ is a small (in norm) perturbation of $H$, so from assumption \ref{tilde gap} it follows that we can find a neighbourhood $\caU = \caU(L)$ of $\phi=0$ such that the twist Hamiltonian $H(\phi)$ also has a non-degenerate ground state, gapped by ${g}/2$. In this neighbourhood $\caU$, we denote by ${P}(\phi)$ the ground state projection of ${H}(\phi)$. We thus have a two-parameter family of projections of which we consider the \emph{adiabatic curvature} at $\phi=0$:
\begin{equation}\label{def:adiabatic curvature}
\kappa  := \iu \;\Tr \left( {P} [\partial_{1} {P}, \partial_{2} {P}] \right) = \iu \; \omega \left( [\partial_{1} {P}, \partial_{2} {P}] \right),    \qquad  \partial_{i}P= \partial_{\phi_i}P(\phi)\big|_{\phi=0}
\end{equation}
We immediately point out that $\kappa$ is independent of the precise form of the vector potential $A$ that was used to define the twist Hamiltonian. Indeed, consider another vortex-free $A'$ that threads the same flux $(\phi_1,\phi_2)$, implying  that it is of the form
$$
A'= \phi_1 (\xi_1 + \dd \theta_1) + \phi_2(\xi_2 + \dd \theta_2) = A+  \dd f,\qquad \text{where}\quad f=  \phi_1 \theta_1 + \phi_2\theta_2.
$$
Then changing $A\to A'$ does not change the adiabatic curvature $\kappa$. 
More precisely,
\begin{lemma} \label{lem: curvature gauge}
Let $H'(\phi)$ be
$$
H'(\phi) = H_{A'}, \qquad  A'= \phi_1 (\xi_1 + \dd \theta_1) + \phi_2(\xi_2 + \dd \theta_2)
$$
for some functions $\theta_{1,2}$ satisfying $\norm{\dd\theta_{1, 2}} \leq C$. By \eqref{eq:H under gauge} these Hamiltonians are also uniformly gapped for $\phi$ in a neighbourhood $\caU'$ of $0$. If we write  $P'(\phi)$ for the corresponding groundstate projections, then
$$
\kappa = \iu \: \omega' \left(  [\partial_1  P', \partial_2 P' ] \right) + \caO(L^{-\infty}).
$$
\end{lemma}
%
%
It is useful to state an alternative, oft-used form of the curvature.  Its basic ingredients  are generators of parallel transport $K_i$. These operators have to satisfy the relation 
\begin{equation}\label{eq: property adiabatic generators}
\partial_{j} P=\iu \: [K_j,P],\qquad j=1,2.
\end{equation}
It is immediate that this relation does not fix $K_i$ uniquely. A choice that one encounters often (but that is rather useless in the many-body setting because it is not local) is $K_i= P (\partial_{i}P) (1-P) + (1-P) (\partial_{i}P) P $. By a little algebra, we see that
\begin{equation} \label{eq: algebra parallel}
  \omega([[K_1,P],[K_2,P]]) =  -\omega([K_1,K_2]) 
\end{equation}
and hence, for any pair  $K_{1,2}$ of generators of parallel transport, 
$$
\kappa= \iu \: \omega([K_1,K_2]).
$$
%
%
%
\section{Results}\label{sec:Results}

To put the results that follow into a firm context, we note that a strong from of quantization was proven in \cite{HastingsMichalakis,bachmann2018many,Giuliani:2016gn} for the adiabatic curvature $\kappa$ as introduced in Section \ref{subsec:adiabatic curvature}, namely,
\begin{theorem} \label{thm: quantization} 
There exists $ n\in \bbZ$ such that
$$
\left| \kappa -2\pi n \right|  =\caO(L^{-\infty}).
$$
\end{theorem}
Note that the proof of this theorem is simplified if one demands that Assumption \ref{tilde gap}  holds for all fluxes $\phi \in \bbT^2$, see \cite{bachmann2018quantization}. In view of this result, we build up the following sections as linking alternatively defined response coefficients to $\kappa$.

\subsection{From the Kubo response to adiabatic curvature}

We want to compute the current density in response to a perpendicular applied electric field. We measure this current density $j_2$ in the $2$-direction and at the origin. The driving is by a uniform electric field of strength $E$ in the $1$-direction. 
The Hall conductivity in this setup should be 
\begin{equation}\label{eq: hall conductivity as ratio}
 \frac{j_2}{E}= \frac{\langle J_{\gamma_d}\rangle}{2dE},\qquad  E\to 0
\end{equation}
where  $\gamma_d$ is the oriented path in $\Gamma^*$  running in the $x_1$-direction from $-d+1/2$ to $d+1/2$ at $x_2=1/2$
 i.e.\ it has length $2d$, see Figure \ref{fig:local_current}. The corresponding current operator $J_{\gamma_d}$ was defined in Section \ref{subsec:current operators}.   To implement the electric field, we choose an electrostatic potential $v$ that gives a constant electric field in the strip $\{|x_1| \leq \ell\}$:
\begin{equation}\label{eq:current within field}
{\dd} v = E   \dd x_1 \quad  \text{ on} \quad  \{|x_1| \leq \ell\}^e
\end{equation}
 \begin{wrapfigure}{r}{0.3\textwidth}
  \begin{center}
    \includegraphics[width=0.3\textwidth]{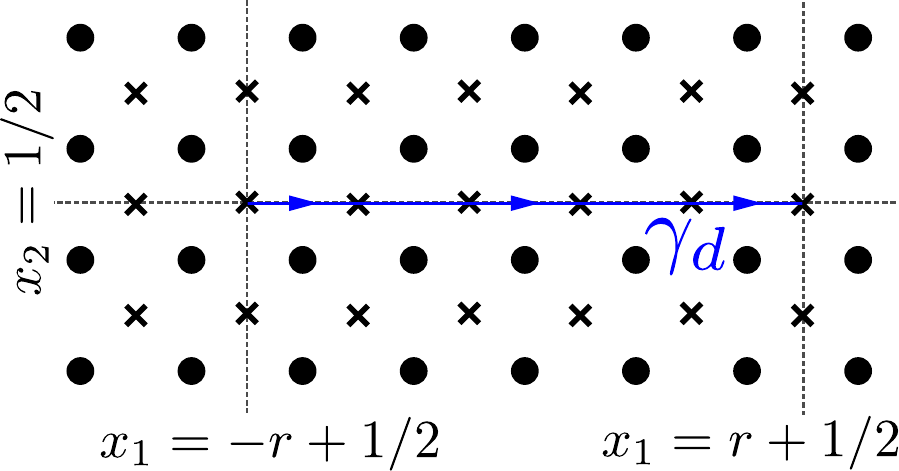}
  \end{center}
  \caption{}
  \label{fig:local_current}
\end{wrapfigure}
with $\ell \geq d$. (one could think that $\ell\gg d$ is necessary but that does not make any difference for the upcoming result) For the rest, $v$ is arbitrary but such that $\norm{\dd v} \leq C$.   The operator implementing this potential is $V=\langle v, n\rangle$ and so we have specified both $J=J_{\gamma_d}$ and $V$, see Figure \ref{fig:lemma42torus}.
\begin{lemma}\label{thm:1}
With $V,J$ chosen as in the lines above, we have
\begin{equation}
\left| \kappa - \frac{\chi_{J,V}}{2Ed} \right| = \caO(1/d).
\end{equation}
\end{lemma}
The relatively large error $\caO(1/d)$ in this theorem is explained by realizing that the current operator $J_{\gamma_d}$ itself is only defined unambiguously up to terms of norm unity at the edges of the line segment, see Section \ref{subsec:current operators}.  This also shows the way to a solution: We note that  \eqref{eq: hall conductivity as ratio} also equals $\frac{\langle J_{\gamma_d}\rangle}{\Delta v}$ with $\Delta v$ the change in potential along the line segment (in other words: for transverse conductivity in 2D, \emph{conductivity} equals \emph{conductance}). 
One is tempted to modify the setup so that the endpoints of $\gamma_d$ are in a field-free region.  Here is a possible way: We keep the electric field the same as before in the strip $\{|x_1| \leq \ell\}^e$ and we insist that it is identically zero in the strips $\{ \ell < |x_1| <\ell+2r \}^e$ with $r\gg 1$. The length $2d$ of path $\gamma_d$ is now chosen $d=\ell+r$, see Figure \ref{fig:theorem43torus}.
So, to nail down the model precisely, we take $J=J_{\gamma_{d}}$ (defined above Lemma \ref{thm:1}) with $d=\ell+r$  and $V=\langle v, n \rangle$ with
\begin{equation}\label{eq:potential outside current}
{\dd} v =  \begin{cases} E  \dd x_1  & \text{on} \quad \{ |x_1|  \leq \ell\}^e  \\
0 & \text{on}  \quad \{  \ell <|x_1| \leq \ell+2r\}^e
 \end{cases}
\end{equation}
and $v$ arbitrary elsewhere but with $\norm{\dd v} \leq C$. We write $\Delta v=v(\ell, 0)-v(-\ell, 0) $
\begin{theorem}\label{thm:accurate}
With $V,J$ chosen as in the lines above
\begin{equation}
\left| \kappa - \frac{\chi_{J,V}}{\Delta v} \right| = \caO \big(r^{-\infty} \big).
\end{equation}
\end{theorem}
  As anticipated, the above result has a much better accuracy than Lemma \ref{thm:1}. 
 What is however not yet explicitly exhibited, is the topological nature of the response coefficient.  We still have a relevant region of constant electric field. However, we note that the error term only depends on $r$ and not  on $\ell, d$ separately: the entire potential difference can also be realized along a single site spacing ($\ell=1$). This already shows that it is not important to have a region around where the electric field is well-defined. 
 We can take this a step further and cast the result in a much more robust way. 
 Let us deform the path $\gamma$, allowing it to be an arbitrary path in $\Gamma^*$ that is part of the oriented boundary of some set (cf. \eqref{eq: current as region current}).
 We denote the begin-and endpoints of $\gamma$ by $y_{\rm b},y_{\rm e}\in \Gamma^*$, thus also specifying an orientation for $\gamma$. 
 We now consider a potential $v$ that is flat on spheres of radius $r$ around\footnote{To make this intuitive condition precise, we refer to the natural embedding of both $\Gamma$ and $\Gamma^*$ in the continuous torus, i.e.\ $[-L/2,L/2]^2$ with edges identified.}  $y_{\rm b}$ and $y_{\rm e}$, see Figure \ref{fig:lemma44torus}.
Abusing the notation slightly, we denote by
$v(y_{\rm b}), v(y_{\rm e})$ the two values that $v$ takes in the vicinity of $v(y_{\rm b}), v(y_{\rm e})$. 
 We then define the potential difference 
 $\Delta v= v(y_{\rm e})-v(y_{\rm b})$.
 We set $J = J_{\gamma}$, then 
 \begin{lemma}\label{thm:accurate even more}
With $V,J$ chosen as in the lines above
\begin{equation}
\left| \kappa - \frac{\chi_{J,V}}{\Delta v} \right| = \caO \big(r^{-\infty} \big).
\end{equation}
\end{lemma}
This lemma is our most revealing result on the Kubo response. 

\begin{figure}
\centering
\captionsetup[subfigure]{width=0.8\imagewidth,justification=raggedright}%
  \begin{subfigure}[b]{0.3\textwidth}
    \includegraphics[width=\textwidth]{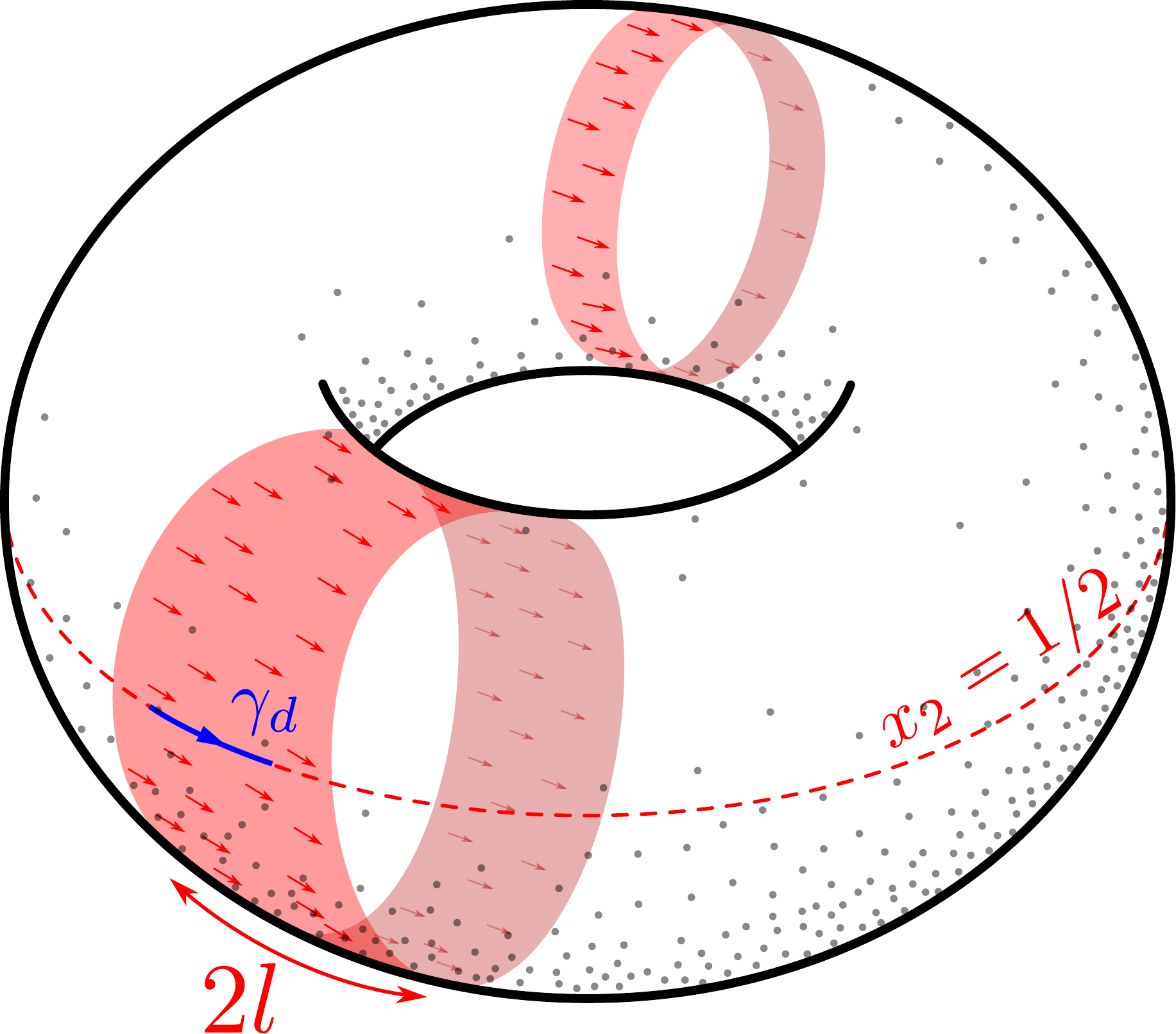}
    \caption{Setup for Lemma \ref{thm:1}. The current is measured across a small segment in the bulk of the electric field. As indicated, since the field $\dd v$ is derived from a potential, it must be non-zero somewhere outside the region $\{ \abs{x_1} \leq l \}$ as well.}
    \label{fig:lemma42torus}
  \end{subfigure}
  \begin{subfigure}[b]{0.3\textwidth}
    \includegraphics[width=\textwidth]{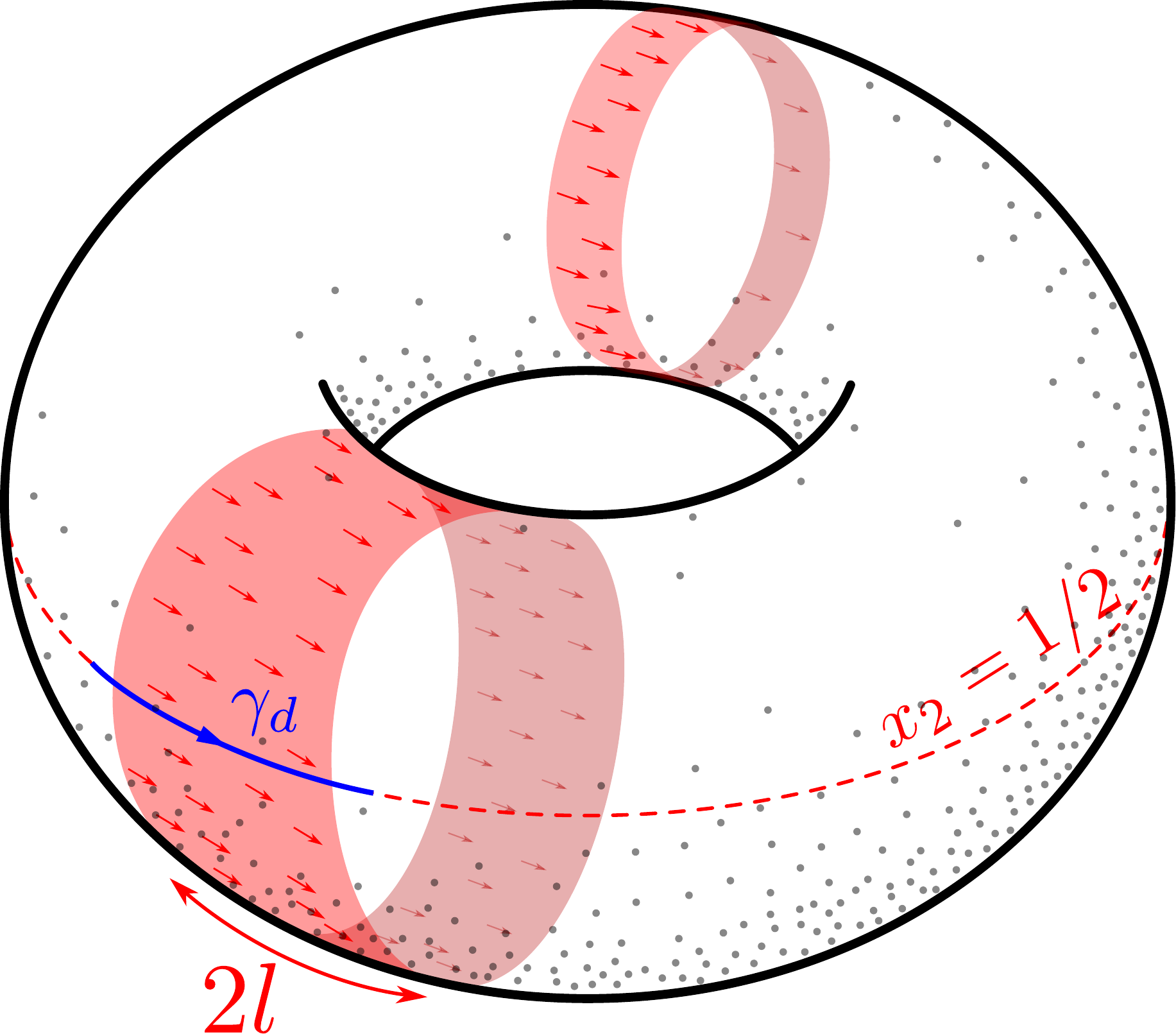}
    \caption{Setup for Theorem \ref{thm:accurate}. The current is measured across a line that completely traverses the region of electric field. \vspace{2cm}}
    \label{fig:theorem43torus}
  \end{subfigure}
  \begin{subfigure}[b]{0.3\textwidth}
    \includegraphics[width=1.1\textwidth]{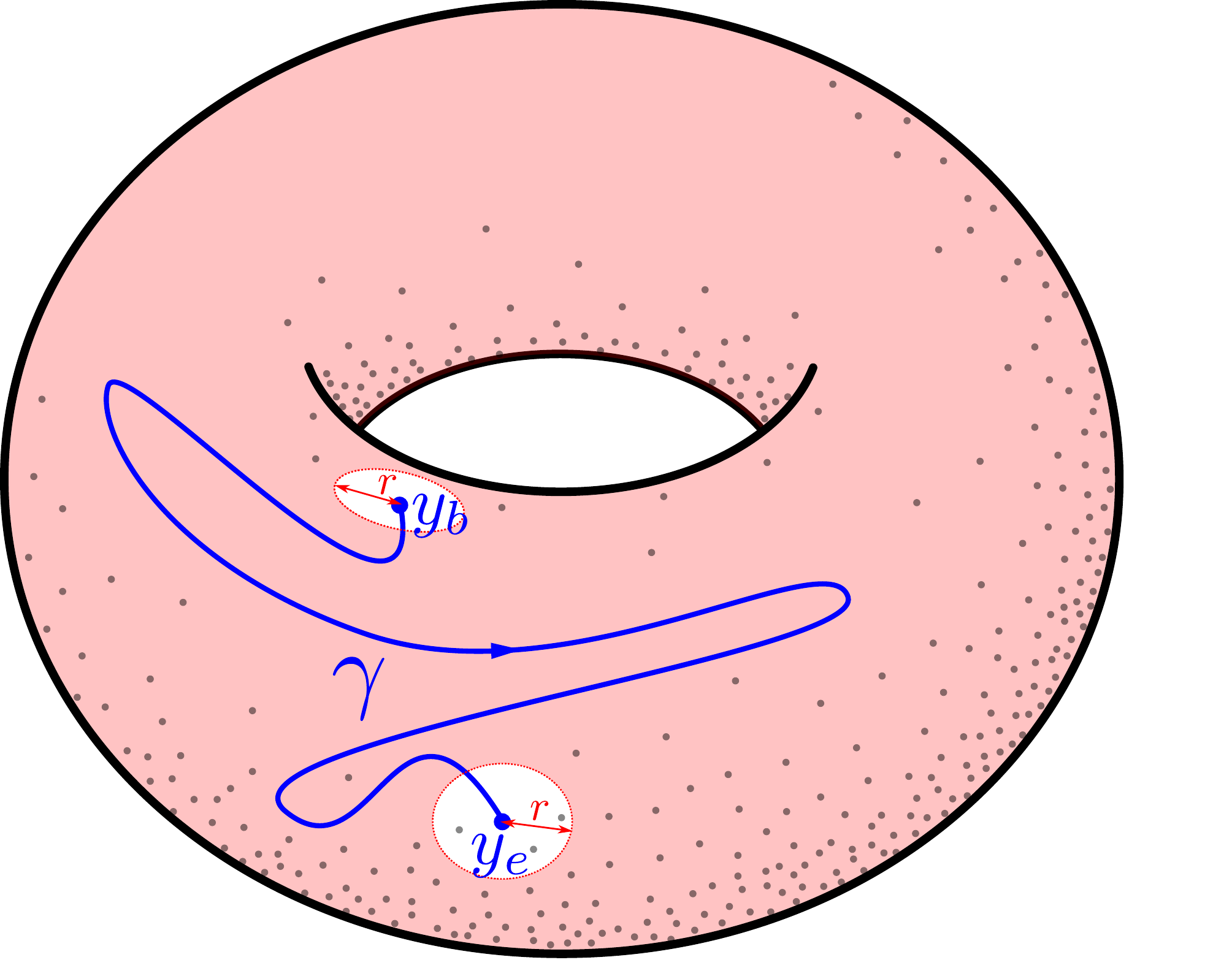}
    \caption{Setup for Lemma \ref{thm:accurate even more}. The electric field $\dd v$ is nonzero only in the red region. \vspace{2.45cm}}
    \label{fig:lemma44torus}
  \end{subfigure}
\end{figure}

\subsection{From Kubo response to adiabatic response}

The setup of adiabatic response demands that we specify a slow change in the Hamiltonian. In the context of Hall fluids, the natural change is to slowly thread a flux.  Hence, we take up the setup introduced in Section \ref{subsec:ham with flux}, we choose a vortex free vector potential $A$ and define 
$$
H_s :=H_{sA}
$$
As we saw in Section \ref{subsec:adiabatic driving} a special role is played by the derivative $W:=\partial_s H_s \big|_{s=0}$. In our case, this derivative is locally computed to be 
\begin{equation} \label{eq: gauge in adiabatic response}
(W)_{\Sigma}= \iu \: ([\langle \theta,n\rangle, H]))_{\Sigma}
\end{equation}
where $\Sigma$ is a region in which $A$ is exact, i.e.\ $A=\dd \theta$ on $\Sigma^e$, see Section \ref{subsec:ham with flux}.  
The commutator in the right-hand side of the previous formula reminds us of the frequency derivative linking the adiabatic and Kubo responses.
If $J$ has support in the far interior of $\Sigma$, we can pretend that \eqref{eq: gauge in adiabatic response} holds globally, leading to
\begin{theorem} \label{thm: link adiabatic and kubo}
Let $J$ be supported in $X$, such that
 $\mathrm{dist}(X,\Gamma \setminus\Sigma) \geq r$, with $A=\dd\theta$ on $\Sigma^e$.  Then 
$$
\left| \chi^{\mathrm{ad}}_{J,H_s} - \chi_{J,V} \right|  = \caO(r^{-\infty}),\qquad \text{with}\quad  V=  \langle \theta,n\rangle.
$$ 
\end{theorem}
The above theorem tells us that adiabatically switching on a vector potential evokes the same response as driving with an electric field (derived from the electrostatic potential $\theta$). 
This is demystified by recalling the standard electrodynamics relation
$E={\dd} v -\partial_t A$ and noting that we have here an $A$ that is linear in the rescaled time $s=\epsilon t$, and the observable $J$ allows to restrict to a region where $A={\dd} \theta$. Hence $E=-{\dd} \theta$ and $\theta$ plays the role of an electrostatic potential $v$. This is precisely the content of the above theorem. 

To belabour this point, we provide a corollary to Theorem \ref{thm: link adiabatic and kubo} that applies to a Hall setup.   Consider a path $\gamma$ that has the regularity also required in Lemma \ref{thm:accurate even more} ( \ie $\gamma$ is part of the boundary of some set) but we allow for the path to be closed as well. 
We consider a vortex-free vector potential $A$ that vanishes in the balls of radius $r$ around the points $y_b$ and $y_e$ (for closed paths, there is no requirement, and then we formally take $r=\caO(L)$). 
Define $\caE:=\int_\gamma A $ (the suggestion is that this is an \emph{emf}, ie.\ electromotive force)
%
\begin{corollary}
Let $H_s=H_{sA}$ with $A$, $J=J_\gamma$ and $\caE$ as described above, then we have
$$\left| \kappa - \frac{\chi^{\mathrm{ad}}_{J, H_{s}}}{\caE}\right| = \caO(r^{-\infty}).$$
\end{corollary}
In the case of an open path, this corollary is an immediate consequence of Theorem \ref{thm: link adiabatic and kubo} and Lemma \ref{thm:accurate even more}, as one can always choose a gauge $\theta$ locally so that $A=\dd \theta$. For closed paths, this might be impossible. In that case one can for example follow the steps of the proof of Thoerem \ref{thm:accurate}, or, alternatively,  still use  Theorem \ref{thm: link adiabatic and kubo} for several paths glued together in regions of diameter $cL$ where $\dd A$ vanishes.

%
%


\section{Proofs}

%

%
\subsection{Preliminaries} \label{sec: preliminaries}
Let $W \in L^\infty(\bbR)\cap L^1(\bbR)$ be an odd function such that 
\begin{enumerate}
\item $
\left\vert W(t)\right\vert = \caO(|t|^{-\infty})$
\item $\widehat W(\zeta) = \frac{-\iu}{\sqrt{2\pi} \zeta}, \quad \text{if } \abs{\zeta} \geq  g/2.$ 
\end{enumerate}
where $\widehat{W}$ is the Fourier transform of $W$.  See  \cite{HastingsWen,Sven} for a construction of such $W$.  Then we define the map $\caI$ (acting on operators $O$)
\begin{equation}\label{def: cai}
\caI(O)=\caI_{H}(O) := \int_{-\infty}^\infty \dd t \; W(t) \ep{\iu t H} O \ep{-\iu t H}.
\end{equation}
Furthermore, we need the \emph{off-diagonal projection} 
$$
O \mapsto \bar{O}=P O P^\perp + P^\perp  O P,\qquad  P^\perp =1-P
$$
where we recall that $P$ is the (one-dimensional) ground state projection of $H$. 
We summarize the useful properties of these objects.
\begin{lemma}\label{lem:I is inverse}
Let $O,O'$ be arbitrary operators. We write $\mathrm{ad}_H(O)=[H,O]$. 
\begin{enumerate}
\item $ \omega(OO')=\omega(\bar{O}O')=\omega({O}\bar{O'})=\omega(\bar{O}\bar{O'})$.
\item  $\overline{\caI(O)}={\caI( \bar O)}$.
\item   $\mathrm{ad}_H \,\caI(\bar O) = \iu \bar O$. 
\item $\caI \, \mathrm{ad}_H( \bar O) = \iu  \bar O$. 
\item If $O$ has support in $S$ then $\norm{\caI(O)-(\caI(O))_{S^r}} =   \norm{O} \abs{S} \times  \caO(r^{-\infty})$. 
\item $\norm{\caI(O)} \leq \norm{W}_1 \norm{O}$.
\end{enumerate}
\end{lemma}
\begin{proof}
We view the algebra of operators as a Hilbert space with the Hilbert Schmidt scalar product (remember that all is finite-dimensional).  This makes $\mathrm{ad}_H$ into a Hermitian operator and we define $\widehat{W}(\mathrm{ad}_H)$ by spectral calculus. 
From \eqref{def: cai}, we see that $\caI=\sqrt{2\pi}\widehat{W}(\mathrm{ad}_H)$.  This proves that $\caI$ and $-\iu \mathrm{ad}_H$ are inverses on the spectral subspace $ |\mathrm{ad}_H | \geq g/2 $. By the gap assumption, this subspace contains all $\bar{O}$.  Hence $(iii), (iv)$ are shown.  The claim $(v)$ follows by the Lieb-Robinson bound and the remaining claims are obvious.  
\end{proof}
\begin{lemma}\label{lem: adiabatic generators}
Consider vector potentials $A=A(\phi_1,\phi_2)$ threading fluxes $(\phi_1,\phi_2)$ as defined in Section \ref{subsec:adiabatic curvature}. Then 
$$K_{j}=\caI(\partial_{\phi_j} H_{A}), \qquad j=1,2$$
(with derivatives taken at $\phi=0$)
are generators of parallel transport, i.e.\ they satisfy \eqref{eq: property adiabatic generators}. 
\end{lemma}
For the proofs, see \cite{HastingsWen,Sven} for the case of spin systems and \cite{bru2016lieb,Teufel17,BrunoInPrep} for fermionic systems. 
\begin{lemma}\label{lem: locality with cai}
 Let $G,G'$ be local Hamiltonians in the sense of Section \ref{sec: spaces}. Let $Z$ be the intersection of their supports. Then 
 $$
 [\caI(G),\caI(G')]=   [\caI(G_{Z^r}),\caI(G'_{Z^r})] +\caO(r^{-\infty}),\qquad    [\caI(G),G']=   [\caI(G_{Z^r}),G'_{Z^r}] +\caO(r^{-\infty})
 $$
\end{lemma}
\begin{proof}
We split the local Hamiltonians in local terms and use Lemma \ref{lem:I is inverse} (v) and (vi).
\end{proof}
\subsection{Proof of Lemma  \ref{lem: existence of lr}: Thermodynamic limit}  \label{sec: proof of thermo}

Let us denote the quantity in 
\eqref{eq:linear response coefficient} without limits as
$$
 \chi(\epsilon,L) :=  \iu \: \int_0^{\infty} \dd t \;  \omega \left( [V(-t), J] \right) \ep{\iu \nu t-\epsilon t}, 
$$
dropping hence $V,J,\nu$ from the notation.  We keep in mind that $J,V$ are independent of $L$ (see Section \ref{sec: spaces}). 
We now proceed in three steps.
\begin{lemma} \label{lem: l to infty}
For any $\epsilon>0$, the following exists
$$
\chi(\epsilon,\infty) :=\lim_{L\to\infty}\chi(\epsilon,L)
$$
\end{lemma}
\begin{proof}
Indeed, for any finite $t$, the $\lim_L \omega \left( [V(-t), J] \right)$ exists by Assumption \ref{tl limit} and locality of dynamics (Lieb-Robinson bound), and it is bounded by $\norm{V}\norm{J}$. Consequently, the limit of the $t$-integral exists by dominated convergence. 
\end{proof}
We now state a lemma that expresses the main point, in the sense that one should not expect it to be true if the system were not gapped.
\begin{lemma}\label{lem: ep to zero}
The limit $\chi(L)= \lim_{\epsilon \downarrow 0} \chi(\epsilon,L)$ exists and (for some $L$-independent $C$)
$$
|\chi(\epsilon,L)-\chi(L)| \leq C \epsilon,\qquad \text{for $\epsilon\leq g$}
$$
\end{lemma}
\begin{proof}
Computing
$$
\int_0^{\infty} \dd t \;\overline V(-t) \ep{\iu \nu t - \epsilon t} = -\iu \; \left( \frac{1}{-(H + \nu) - \iu \epsilon} P^{\perp} V P + P V P^{\perp} \frac{1}{(H - \nu) - \iu \epsilon} \right)
$$
and using Lemma \ref{lem:I is inverse} (i) we find
$$
\chi(\epsilon, L) = \omega \left( V P^{\perp} \frac{1}{(H - \nu) - \iu \epsilon} P^{\perp} J \right) - \omega \left( J P^{\perp} \frac{1}{-(H + \nu) - \iu \epsilon} P^{\perp} V \right).
$$
Since $\nu$ is smaller than $g/2$, half the gap of $H$, the limit is obviously the same expression with $\epsilon=0$ and the difference from the limit is, by functional calculus, bounded by ($\pm \nu$ corresponding to the two terms above)
$$
2\norm{V}\norm{J} || {P^\perp}  \frac{\epsilon}{(H\pm\nu)(H\pm\nu-\iu \epsilon) }  {P^\perp}  ||  \,  \leq \,  \frac{C\epsilon}{g} (1+\frac{4\epsilon^2}{{g}^2} )  \norm{V}\norm{J}.
$$

\end{proof}

\begin{lemma}\label{lem: ep then l to zero}
The limit $\lim_{L\to\infty}\chi(L)$ exists.
\end{lemma}
\begin{proof} We use the language of Section \ref{sec: preliminaries}, in particular we consider the operator
$\mathrm{ad}_H$ acting on a Hilbert space. Since the spectrum of $\mathrm{ad}_H + \nu$ contains no points other than zero that are smaller than $g/2$ (remember that $\abs{\nu} \leq g/2$), we find that
$$
\chi(L) = \iu \, \int \dd t \; W(t) \; \ep{\iu t \nu} \omega([{V}(-t), J]).
$$
with the function $W$ defined in Section \ref{sec: preliminaries}. The operator $\int \dd t \; W(t) \; \ep{\iu t \nu} {V}(-t)$ can be well-approximated by local operators, by the same reasoning as in the proof of Lemma \ref{lem:I is inverse} ($v$). The claim consequently follows by Assumption \ref{tl limit} and dominated convergence.
\end{proof}

Lemma \ref{lem: existence of lr} now follows directly by combining Lemmata \ref{lem: l to infty}, \ref{lem: ep to zero} and \ref{lem: ep then l to zero}. 

\subsection{Proof of Lemma \ref{lem: locality of lr}  }

In the course of the proof in Section \ref{sec: proof of thermo}, we have in particular obtained
\begin{equation}
\chi_{J,V} = \iu \, \omega([\caI( V),J]).
\end{equation}
The locality now follows directly from Lemma \ref{lem: locality with cai}.

\subsection{Proof of Lemma~\ref{mhw}} \label{sec: proof of derivative}

Starting from (\ref{def:adiabatic response}), it was shown in \cite{bachmann2017adiabatic} that
$$
\chi^{\mathrm{ad}}_{J,H} = \iu \, \omega( [\caI (K), J]),\qquad K=K_{s=0}.
$$
 We now connect RHS of (\ref{eq:LR form of adiabatic response}), lets call it $\chi$, to this expression. Lemma~\ref{lem:I is inverse} says that the operation $\caI$ is an inverse of $- i \mathrm{ad}_H$ when restricted to an appropriate space. In particular using points (i) and (iv) of the lemma we get that 
$$
\omega([\caI (O(-t)), O'])
$$
is a primitive function of $\omega([O(-t), O'])$ for any observables $O, O'$. Integrating the expression (\ref{eq:LR form of adiabatic response}) for $\chi$ by parts we obtained
$$
\chi = - \iu \lim_{\epsilon \to 0^+} \int_0^\infty \dd t   \; (1 - \epsilon) \ep{-\epsilon t} \omega \left( [\caI(W(-t)), J] \right).
$$
By the same arguments that were used to prove the existence of thermodynamic limit, the part with $\epsilon$ vanishes in the limit. Noting that $\caI(W) = K$ and integrating by parts again we get
$$
\chi = \iu \lim_{\epsilon \to 0^+} \omega \left( [\caI(K), J] \right) - \iu \lim_{\epsilon \to 0^+} \int_0^\infty  \dd t   \; \epsilon \, \ep{-\epsilon t} \omega \left( [\caI(K(-t)), J] \right).
$$
The second part again vanishes in the limit and we obtain $\chi = \chi^{\mathrm{ad}}_{J,H}$. 


\subsection{Proof of Lemma \ref{lem: curvature gauge} }

From ~\eqref{eq:H under gauge} we see that the projections $P'(\phi)$ are related to $P(\phi)$ through the gauge transformation $U(\phi) = \ep{\iu \langle \phi_1 \theta_1 + \phi_2 \theta_2, n \rangle}$, therefore
$$
\partial_{1, 2} P'(\phi) = U(\phi) \big( \iu [ \langle \theta_{1, 2}, n \rangle, P(\phi)] + \partial_{1, 2} P(\phi) \big) U(\phi)^*.
$$

Let's write ${V}_{1, 2} = \langle \theta_{1, 2}, n \rangle$, then the adabatic curvature for the family of projections $P'(\phi)$ is (all derivatives at $\phi=0$)
\begin{align*}
\iu \; \omega' \left( [\partial_1 P', \partial_2 P'] \right) &= \iu \; \omega \left( [\partial_1 P, \partial_2 P] \right) \\
&- \omega \left( [[{V}_1, P], \partial_2 P] \right) \\
&- \omega \left( [\partial_1 P, [{V}_2, P]] \right) \\
&- \iu \; \omega \left( [[{V}_1, P], [{V}_2, P]] \right).
\end{align*}

The first term on the right-hand side is the adiabatic curvature $\kappa$ of the family $P(\phi)$, it remains to show that the other three terms vanish. The fourth term is $\iu \, \omega \left( [{V}_1, {V}_2] \right)=0$ by the same algebra as in \eqref{eq: algebra parallel}, because $[{V}_1, {V}_2]=0$. We show now why the second term vanishes up to $\caO(L^{-\infty})$ (the third term is analogous).  
We have
\begin{equation}\label{eq: proof of gauge decomp}
\omega \left( [[{V}_1, P], \partial_2 P] \right) = \iu \; \omega \left( [{V}_1, \caI \big( \partial_2 H_A \big)] \right) = \iu \sum_{x}  \theta_1(x) \; \omega \left( [n_x, \caI \big( \partial_2 H_A \big)] \right).
\end{equation}
For each $x$, we consider a region $\Sigma^x$ of diameter $\caO(L)$ centered on $x$. In this region, we have $A=d(\phi_1 f_1+\phi_2f_2)$ for some $f_{1,2}$ and  hence 
$$
(\partial_2 H_A \big)_{\Sigma^x}= ([\langle f_2,n\rangle, H \big)_{\Sigma^x}.
$$
In the last expression, we changed $H_A\to H$ as the derivative was at  $\phi=0$. 
Because of locality of $\caI$ and the boundedness of $A$, we have
$$
[n_x, \caI \big( \partial_2 H\big)]=  [n_x, \caI \big( [\langle f_2,n\rangle, H]\big)] + \caO(L^{-\infty}).
$$
Now, 
$$
\omega \left( [n_x, \caI \big( [\langle f_2,n\rangle, H]\big)]  \right) =  \omega \left( [n_x,  {{\langle f_2,n\rangle}}]\right)  =0 
$$
where we used Lemma \ref{lem:I is inverse} (i), (ii) and (iv). The claim is proven by plugging this into \eqref{eq: proof of gauge decomp}.

\subsection{Proof of Theorem \ref{thm:accurate} }

We start from 
\begin{equation}
\chi = \iu \, \omega([\caI(\bar V),J]),
\end{equation}
Because the function $W$ in the definition of $\caI$ is odd, we have also 
\begin{equation}\label{eq:linear response intermediate}
\chi = -\iu \, \omega([\bar V,\caI(J)]).
\end{equation}
Using Lemma \ref{lem:I is inverse} (iv), we then obtain
\begin{equation}\label{eq: double i}
\chi  =  -\omega([\caI([H,\bar V]),\caI(J)])= -\omega([\caI([H, V]),\caI(J)]) 
\end{equation}
The intersection of the supports of $[H,V]$ and $J$ is contained in 
$$Z:= \{ |x_1| \leq \ell+C, |x_2| \leq C \},$$
see Figure \ref{fig:thm43_figure}. Since $[H,V]$ and $J$ are clearly `local Hamiltonians', we can apply Lemma \ref{lem: locality with cai} to conclude that \eqref{eq: double i} equals
\begin{equation} \label{eq: first app of cai}
[\caI([H, V]),\caI(J)] = [\caI([H, V]_{Z^r}),\caI(J_{Z^r})] +\caO(r^{-\infty}).
\end{equation}
Now, let us approach from a different angle and consider the vector potential 
$$
A=\phi_1 A_1+\phi_2A_2
$$
where 
\begin{enumerate}
\item $A_1=dv$ on $\{|x_1| \leq \ell+2r\}^e$ and $A_1=0$ elsewhere. Here $v$ was defined just above Theorem \ref{thm:accurate}. 
\item $A_2=d h_2$ on $\{ |x_2 \leq C|\}^e$ and $A_2=0$ elsewhere, with $h_2$ the Heaviside function $h_2(x)=1(x_2>0)$.
\end{enumerate}
The most relevant properties of $A$ are that
\begin{equation} \label{eq: identification a}
[H, V]_{Z^r} =   \iu ( \partial_{\phi_1} H_{ A})_{Z^r},\qquad  J_{Z^r}= (\partial_{\phi_2} H_{ A})_{Z^r}
\end{equation}
see \eqref{eq: ha locally gauge} and \eqref{eq: current as region current}. 
Additionally, the intersection of the supports of $\partial_1 H_A$ and $\partial_2 H_A$ (derivatives at $\phi=0$ is also contained in $Z$ and these are also local Hamiltonians, so Lemma \ref{lem: locality with cai} applies here as well.  Combining this fact with \eqref{eq: first app of cai} and \eqref{eq: identification a}, we conclude that
$$
[\caI([H, V]),\caI(J)]=  \iu  [\caI((\partial_{\phi_2} H_{ A})),\caI((\partial_{\phi_1} H_{ A}))]  +\caO(r^{-\infty}) 
$$

\begin{figure}[h]
\centering
\includegraphics[width=0.2\textwidth]{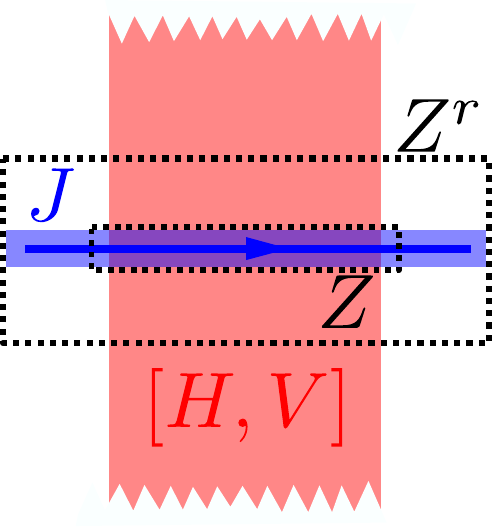}
\caption{The supports of $[H, V]$ and $J$, and the sets $Z$ and $Z^r$ in a neighbourhood of the location where the current is measured.}
\label{fig:thm43_figure}
\end{figure}

%
%
%
%
%
%
The expression on the right is almost of the type as appeared in the definition of adiabatic curvature, except that there we demanded that $A$ threads fluxes $\phi_1,\phi_2$. In our situation, $\phi_1 A_1$ threads a flux $\phi$, but $\phi_2 A_2$ threads a flux $2\phi_2 \ell E = \phi_2 \Delta v$.  This shows that the $\omega(\cdot)$ of the above commutator is given by 
$
 (\Delta v) \kappa
$ instead of $\kappa$. This proves Theorem \ref{thm:accurate}.  

\subsection{Proof of Lemma \ref{thm:accurate even more}}

The same as above, but with different vector potentials $A_1',A_2'$ that are however related to $A_1,A_2$ by gauge transformations.

\subsection{Proof of Lemma \ref{thm:1}}

Lemma \ref{thm:1} uses $\ell$ and $d\leq \ell$. We prove the lemma for $d=\ell$. This suffices because, if $d < \ell$ then we modify the potential $v$ by making it flat for $|x_1|\geq d$. By the locality estimate \ref{lem: locality of lr}, this changes the response coefficient by $\caO(1)$, which is compatible with the claim of the lemma. 
Now to the argument for $d=\ell$.  Theorem \ref{thm:accurate} applies to our situation, with the modification that the path $\gamma$ in  $J_{\gamma}$ has length $2(\ell+r)$, whereas we need a shortened path $\gamma'$ of length $2\ell$.  However,  $J_{\gamma'}=J_{\gamma} +\caO(r) $ and this difference gives a contribution of order $\caO(r)$ in the response coefficient. This follows indeed from the representation in \ref{lem: locality of lr} and the bound in Lemma \ref{lem:I is inverse}(vi).  Upon division by $\ell$ we get the desired claim.

\subsection{Proof of Theorem \ref{thm: link adiabatic and kubo}}

We start from the expression (Section \ref{sec: proof of derivative})
$$
\chi^{\mathrm{ad}}_{J,H} = \iu \, \omega(  [\caI (K), J]),\qquad K=K_{s=0}.
$$
By Lemma \ref{lem: adiabatic generators} and the definition of $W$, we have  $K=\caI(\partial_s H_s |_{s=0})= \caI(W) $, so that
$$
\chi^{\mathrm{ad}}_{J,H} =  \iu \, \omega(  [\caI (\caI(W)), J])= - \iu \, \omega(  [ \caI(W),\caI( J)])
$$
From \eqref{eq: gauge in adiabatic response}, we know that  $W_\Sigma=\iu [V,H]_\Sigma$. Since the observable $J$ is supported far from $\Lambda \setminus \Sigma$, we invoke  Lemma \ref{lem: locality with cai} to get 
$$
\chi^{\mathrm{ad}}_{J,H} =  \omega(  [ \caI( [V,H]),\caI( J)]) +\caO(r^{-\infty})
$$
By Lemma \ref{lem:I is inverse} $(i),(ii),(iv)$, the right-hand side equals $-\iu \, \omega(  [ V,\caI( J)])= \iu \, \omega(  [ \caI(V),J]) $, which was to be proven.

\section{Appendix}

We provide the necessary definitions for the framework of discrete one-forms $A$ on $\Gamma$. 

\subsection{The vector field of one-forms}
Let $\Gamma^{e} := \{  (x, y) \in \Gamma^2 \; : \; x \sim y \}$ be the set of oriented edges of $\Gamma$. For any oriented edge $(x, y)$, let $\overline{(x, y)} = (y, x)$ be the reversed edge. A \emph{one-form} is a function $A : \Gamma^{e} \rightarrow \bbR$ such that $A(e) = -A(\overline e)$.

%

\subsection{Integration of one-forms along paths}

For an oriented edge $(x, y)$ we define $i((x, y)) = x$ and $f((x, y)) = y$. An \emph{oriented path} in $\Gamma$ is an ordered set of oriented edges $\gamma = (e_1, \cdots, e_N)$ such that $f(e_i) = i(e_{i+1})$ for $i = 1, \cdots N-1$. The \emph{integral of the one-form $A$ along the oriented path $\gamma$} is defined by
\begin{equation}
\int_{\gamma} A := \sum_{e \in \gamma} A(e).
\end{equation}


\subsection{Contractible loops and vortex free one-forms}

For any path $\gamma = (e_1, \cdots, e_N)$ we write $i(\gamma) = i(e_1)$ for the startingpoint and $f(\gamma) = f(e_N)$ for the endpoint of the path. A \emph{loop} is a path $\gamma$ for which $i(\gamma) = f(\gamma)$. We wish to classify loops as `contractible' or `non-contractible' in such a way that we recover the usual homology of the two-torus\footnote{Strictly speaking, the contractible loops give the first homotopy of the space, while we are interested in the homology. The natural setup to discuss homology is to work with simplicial complexes and $r$-chains. The first homology is then characterized by the 1-chains that have no boundary and are not the boundary of some 2-chain. Closed loops are very much like 1-chains without boundary, and being contractible implies being the boundary of a 2-chain. It is therefore clear that the non-contractible loops capture enough information to describe the homology of the torus.}.

One way of doing this is to think of the discrete torus $\Gamma$ as a subset of a smooth flat torus $\bbT^2$. We associate to each edge $(x, y)$ of $\Gamma$ a curve tracing the shortest path from $x$ to $y$ in the torus $\bbT^2$. To each path $\gamma$ we associate the curve obtained by concatenating the curves associated to the edges of $\gamma$. I this way, a closed curve in $\bbT^2$ is associated to each loop in $\Gamma$. We say that the loop $\gamma$ is contractible if its associated curve is contractible in $\bbT^2$.

A one-form $A$ is \emph{exact} in the region $\Sigma \subset \Gamma$ if $\oint_{\gamma} A = 0$ whenever $\gamma$ is a contractible loop in $\Sigma$.

Let $\theta : \Gamma \rightarrow \bbR$, then we define its exterior derivative to be
\begin{equation}
\dd \theta \big( (x, y) \big) = \theta(y) - \theta(x).
\end{equation}
$\dd \theta$ is exact in any subset of $\Gamma$, the integral of $\dd \theta$ vanishes along \emph{all} loops, even the non-contractible ones.

Conversely, if $A$ is exact in the region $\Sigma$, then there exist a function $\theta : \Sigma \rightarrow \bbR$ such that $A \big|_{\Sigma^e} = \dd \theta$. Indeed, pick a point $x_0$ in each connected component of $\Sigma$ and put $\theta(x_0) = 0$. For any other point $x$ that is path-connected to $x_0$, take any path $\gamma$ from $x_0$ to $x$ that lies in $\Sigma$ and define $\theta(x) = \int_{\gamma} A$. This definition is independent of the chosen path because $A$ is exact in $\Sigma$. Now, for any edge $(x, y) \in \Sigma^e$ we have
$$
\dd \theta \big( (x, y) \big) = \theta(y) - \theta(x) = \int_{\gamma_y} A - \int_{\gamma_x} A = \int_{\{(x, y)\}} A = A \big( (x, y) \big)
$$
where $\gamma_x$ and $\gamma_y$ are paths in $\Sigma$ from $x_0$ to $x$ and to $y$ respectively.




\bibliographystyle{unsrt}
\bibliography{Quant.bib}

\end{document}